%% file: template.tex
\documentclass[10pt,journal,twocolumn,twoside]{IEEEtran}
\usepackage{setspace}

\usepackage{cite}
\usepackage{soul}
\usepackage{caption}
\usepackage{array}
\usepackage{amsmath,amssymb,amsfonts}
\usepackage{bm}
\usepackage{algorithmic}
\usepackage{graphicx}
\usepackage{tabularx}
\usepackage{booktabs}
\usepackage{multirow}
\usepackage{float}
\usepackage{textcomp}
\usepackage{xcolor}
\usepackage{amsthm}
\usepackage{subfig}
\usepackage{algorithm}
\usepackage{etoolbox}
\usepackage{cleveref}

\input{defs.tex}

\ifCLASSINFOpdf
\else
\fi
\begin{document}
\title{Multi-Target DoA Estimation with a Single Rydberg Atomic Receiver by Spectral\\
Analysis of Spatially-Resolved Fluorescence}
\author{Liangcheng Han, Haifan Yin,~\IEEEmembership{Senior~Member,~IEEE}, and Mérouane Debbah,~\IEEEmembership{Fellow,~IEEE}%
\thanks{Liangcheng Han and Haifan Yin are with the School of Electronic Information and Communications, Huazhong University of Science and Technology, Wuhan 430074, China (e-mail: hanlc@hust.edu.cn; yin@hust.edu.cn).}%
\thanks{Mérouane Debbah is with KU 6G Research Center, Department of Computer and Information Engineering, Khalifa University, Abu Dhabi 127788, UAE (email:merouane.debbah@ku.ac.ae) and also with CentraleSupelec, University ParisSaclay, 91192 Gif-sur-Yvette, France}%
\thanks{The corresponding author is Haifan Yin.}%
\thanks{This work was supported by the Major Program (JD) of Hubei Province (2025BEA001), the Fundamental Research Funds for the Central Universities and the National Natural Science Foundation of China under Grant 62071191.}}%

\maketitle

\begin{abstract}
Rydberg-based Direction-of-Arrival (DoA) estimation has been hampered by the complexity of receiver arrays and the single-target, narrow-band limitations of existing single-receiver methods. This paper introduces a novel approach that addresses these limitations. We demonstrate that by spatially resolving the fluorescence profile along the vapor cell, the multi-target problem can be effectively solved. Our approach hinges on the insight that by superimposing incoming signals with a strong local oscillator (LO), the complex atomic absorption pattern is linearized into a simple superposition of sinusoids. In this new representation, each spatial frequency uniquely and directly maps to the DoA of a target. This reduces the multi-target challenge into a spectral estimation problem, which we address using Prony’s method. Our approach, termed Imaging-based Spectral Estimation (ISE), inherently supports multi-target detection and restores the full broadband capability of the sensor by removing the restrictive cell-length dependency.  This development also shows potential for realizing multi-channel Rydberg receivers and the continuous-aperture sensing required for holographic multiple-input multiple-output (MIMO). We develop a comprehensive theoretical model, derive the Cramér-Rao Lower Bound (CRLB) as a performance benchmark, and present simulations validating the effectiveness of the approach to resolve multiple targets. 
\end{abstract}

\begin{IEEEkeywords}
Rydberg atoms, Direction-of-Arrival (DoA) estimation, spectral estimation, atomic sensors, multi-target detection.
\end{IEEEkeywords}
\IEEEpeerreviewmaketitle
\section{Introduction}\label{I}
Rydberg atomic receivers are emerging as a transformative technology for wireless communications and sensing, offering unprecedented sensitivity that challenges the limits of classical electromagnetic sensors \cite{gong2025rydberg,schlossberger2024rydberg, zhang2024rydberg, yuan2023quantum}. These sensors leverage the properties of Rydberg atoms—alkali metals like Rubidium (Rb) excited to a high principal quantum number—which exhibit giant dipole moments and a correspondingly strong response to incident radio-frequency (RF) fields \cite{guo2025aoa, fancher2021rydberg, gallagher1994rydberg}. This quantum-based approach enables the measurement of the amplitude, phase, frequency, and polarization  of an RF field by optically probing the energy level shifts within an atomic vapor \cite{holloway2014broadband, sedlacek2013atom, chen2024instantaneous, simons2019rydberg, jing2020atomic}. The key advantages over traditional conductor-based antennas include immunity to thermal noise and ultra-wideband tunability without structural changes, spanning from MHz to THz frequencies \cite{guo2025aoa, holloway2014broadband, simons2021continuous, liu2022continuous}.

A typical application in this domain is  Direction-of-Arrival (DoA) estimation, which is central to the function of radar, navigation, and modern wireless systems \cite{eranti2022overview, skolnik2008radar, evans1982advanced, moghaddasi2020multifunction}. Conventionally, DoA is determined using phased arrays of spatially separated antennas to measure the relative phase of an incoming wavefront \cite{chen2010introduction}. Following this paradigm, initial efforts in Rydberg-based DoA sensing have predominantly relied on arrays of two or more atomic receivers \cite{robinson2021determining, mao2024digital, gong2025rydberg-multi-target}. Recently, a novel method was proposed to achieve DoA estimation using only a \emph{single} Rydberg atomic receiver \cite{guo2025aoa}. This technique ingeniously exploits the interference pattern formed between the incident RF signal and a local oscillator (LO) field within the vapor cell. By measuring the total transmitted power of a probe laser, which is modulated by the spatially varying field intensity along the cell, the DoA of a single incoming signal can be recovered using a Particle Swarm Optimization (PSO) algorithm.

While insightful, the single-receiver method based on integrated power measurement presents certain challenges for broader practical application. First, it is primarily designed for estimating the DoA of a single signal, as the complexity of the optimization landscape increases substantially in multi-signal scenarios. Second, the accuracy of the method is sensitive to a specific vapor cell length that is dependent on the signal's wavelength. This requirement limits the broadband operational capability of Rydberg sensors, as the physical hardware would need to be changed to accommodate different frequencies. To address these issues, alternative single-receiver approaches have been explored. For instance, the work in \cite{chen2025polarization} exploits the intrinsic vector sensitivity of Rydberg atoms to both electric-field polarization and RF magnetic-field orientation, reconstructing the DoA from Zeeman-resolved spectral signatures without requiring spatial diversity.

Inspired by these advances and the core concept of relating DoA to the field distribution within the vapor, we propose a new paradigm that overcomes the aforementioned limitations by using spatially-resolved fluorescence imaging.   The experimental feasibility of this imaging technique is well-established. For instance, methods for observing the fluorescence from a light sheet passed through an atomic vapor have been successfully used to reconstruct 3D beam profiles \cite{radwell20133d}, image standing waves from DC to GHz \cite{schlossberger2024two}, and map magnetic fields \cite{schlossberger2024two}. These works collectively demonstrate that fluorescence imaging is a versatile and powerful tool for spatially resolving phenomena within atomic vapors, providing a robust experimental foundation for  our proposed approach. Building on this, recent work has  shown that the DoA of an RF field can be inferred by imaging standing waves formed by reflections off the dielectric walls of a glass vapor cell \cite{schlossberger2025angle}.
\begin{figure*}[htbp]
\centering
\includegraphics[width=0.8\textwidth]{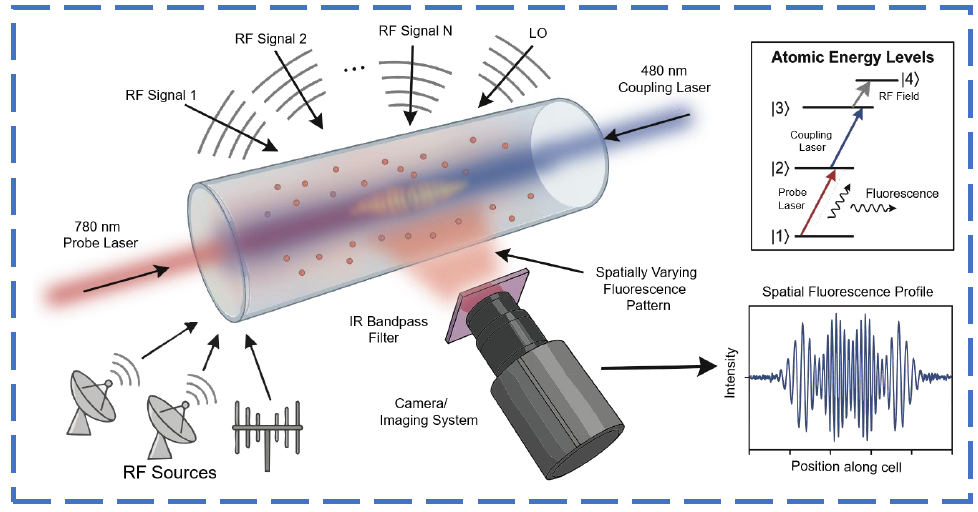}
  \caption{Conceptual diagram of the spatially-resolved Rydberg atomic receiver.}
  \label{fig:rydberg_diagram}
  \vspace{-0.5cm}
\end{figure*}

In this paper, we advance this concept by introducing a novel theoretical framework that enables multi-signal DoA estimation with a single vapor cell. Our approach leverages lateral fluorescence imaging to spatially resolve the RF-field-induced Autler-Townes splitting profile along the propagation axis of the probe laser. As detailed in our theoretical model, the superposition of multiple incident signals and an LO field creates a complex yet coherent spatial intensity pattern within the vapor cell. Under an LO-dominant regime, the resulting local atomic absorption coefficient can be linearized and expressed as a sum of cosines. Each cosine term corresponds to a unique spatial frequency, which is directly related to the direction of arrival of an individual signal. By spatially sampling the fluorescence profile, we transform the DoA estimation problem into a classic spectral estimation problem. For this task, we employ Prony's method. While Prony's method can be sensitive to noise, its application is justified here because Rydberg atomic sensors are fundamentally immune to the thermal noise that plagues conventional receivers \cite{guo2025aoa, gong2024rydberg}. Furthermore, compared to other high-resolution DoA estimation algorithms, Prony's method offers the distinct advantages of requiring only a single snapshot of the spatial data and possessing low computational complexity, making it highly efficient for this application. Importantly, this approach has the potential to provide a unified solution to two major challenges in the field: realizing multi-channel Rydberg receivers \cite{allinson2024simultaneous, meyer2023simultaneous, liu2022deep} and achieving the continuous-aperture sensing with Rydberg atoms required for holographic multiple-input multiple-output (MIMO) \cite{huang2020holographic, gong2023holographic}.
This approach, which we term Imaging-based Spectral Estimation (ISE),  overcomes the shortcomings of the previous integrated-power method. It inherently supports multi-signal detection, removes the restrictive dependency on a specific cell length, and thereby restores the full broadband operational capability of the Rydberg sensor. The main contributions of this paper are as follows
\begin{itemize}
    \item We develop a comprehensive theoretical framework that connects the spatially-resolved fluorescence profile within the vapor cell to the underlying multi-signal RF interference pattern. This model overcomes the single-target limitation of previous single-receiver methods by replacing integrated power measurement with a spatially-dependent analysis.
    \item We demonstrate how the inherently non-linear atomic response to the RF field can be linearized using a first-order Taylor expansion under an LO-dominant approximation. This crucial step reduces the complex multi-signal interaction to a simple superposition of sinusoids, transforming the problem into a solvable linear spectral estimation task.
    \item We introduce a virtual array processing scheme to discretize the continuous fluorescence profile for digital analysis. By applying a family of shifted spatial window functions, this method creates a uniformly sampled measurement vector. The resulting data, a sum of sinusoids, is in an ideal form for high-resolution spectral estimation algorithms like Prony's method.
    \item We derive the Cramér-Rao Lower Bound (CRLB) for multi-target DoA estimation in our system, explicitly accounting for other unknown parameters like signal amplitude and phase. This analysis provides a fundamental performance benchmark and reveals key dependencies between estimation accuracy, signal-to-noise ratio (SNR), and physical parameters such as the vapor cell length.
\end{itemize}

The remainder of this paper is organized as follows. Sec. \ref{II} presents the system model and theoretical principles. Sec. \ref{III} develops the proposed ISE approach, detailing the linearization under LO-dominance, virtual array processing, and Prony's method for spectral estimation. Sec. \ref{sec:crlb} derives the CRLB as a performance benchmark. Sec. \ref{IV} provides simulation results validating the approach. Sec. \ref{sec:discussion} addresses practical implementation considerations, and Sec. \ref{V} concludes the paper.

\section{System Model and Theoretical Principles}\label{II}
This section outlines the fundamental system model and theoretical underpinnings of our proposed Rydberg atomic receiver for multi-target DoA estimation.

\subsection{Rydberg Atom-Based Receiver Architecture}
The receiver system is centered around a glass vapor cell containing a dilute gas of Rb atoms. As illustrated in Fig. \ref{fig:rydberg_diagram}, the sensing mechanism employs a four-level atomic energy scheme. A 780~nm probe laser is tuned to the transition between the ground state $|1\rangle$ and an intermediate state $|2\rangle$, and a counter-propagating 480~nm coupling laser drives the transition from $|2\rangle$ to a highly excited Rydberg state $|3\rangle$. The incident RF field, which is the signal to be measured, is tuned to be near-resonant with the transition between $|3\rangle$ and an adjacent Rydberg state $|4\rangle$.

The probe laser excites atoms to the intermediate state $|2\rangle$, which then spontaneously decay and emit 780~nm fluorescence. The coupling laser establishes an Electromagnetically Induced Transparency (EIT) condition, which ideally makes the vapor transparent to the probe laser, minimizing both absorption and fluorescence. The external RF field perturbs the Rydberg states via the Stark effect, which modifies the EIT resonance condition. This change in the atomic state modulates the susceptibility of the atomic vapor, $\chi$, and, consequently, the intensity of the emitted fluorescence. The output of the sensor is obtained by imaging this spatially-varying fluorescence from the side with a camera. An infrared (IR) bandpass filter is placed before the camera to block any scattered light from the coupling laser. 

\subsection{Field Interference and Optical Response}
The core principle of many Rydberg-based sensing techniques, including the one proposed in this paper, is the creation of a spatially varying interference pattern within the vapor cell. This is achieved by superimposing the incoming signal fields with a known LO field. For a single incoming plane wave signal, $E_{\text{sig}}(x) = A e^{j(kx\sin\theta+\phi)}$, and an LO field, $E_{\text{LO}}(x) = A_0 e^{j(kx\sin\theta_0+\phi_0)}$, the total RF field is $E_{\text{RF}}(x) = E_{\text{sig}}(x) + E_{\text{LO}}(x)$. The intensity of this field, which dictates the atomic response, is given by \cite{guo2025aoa}
\begin{equation}
|E_{\text{RF}}(x)|^2 = A_0^2 + A^2 + 2A_0A\cos(\Delta k x - \Delta\phi),
\label{eq:field_intensity}
\end{equation}
where $\Delta k \triangleq k(\sin\theta_0 - \sin\theta)$ is the spatial frequency of the interference pattern, and $\Delta\phi = \phi - \phi_0$ is the phase difference. The DoA, $\theta$, is uniquely related to this spatial frequency. The local RF field intensity determines the local Rabi frequency, $\Omega_{\text{RF}}(x) = \mu_{\text{RF}}|E_{\text{RF}}(x)|/\hbar$, where $\mu_{\text{RF}}$ is the RF transition dipole moment.

This spatially modulated Rabi frequency, in turn, modulates the optical properties of the vapor. The interaction between the probe light and the atomic vapor is governed by the Beer-Lambert law. Following the micro-element method established in \cite{guo2025aoa}, the total power transmission ratio through a vapor cell of length $L$ is given by
\begin{equation}
\frac{P_{\rm out}}{P_{\rm in}} = \exp\left(-\int_0^L \alpha(x)\,dx\right),
\label{eq:globalT}
\end{equation}
where the local absorption coefficient $\alpha(x)$ is given by
\begin{equation}
\alpha(x) = k_{pr}\,\mathrm{Im}\,\chi(x),
\label{eq:alpha_def}
\end{equation}
where $k_{pr} = 2\pi/\lambda_{pr}$ denotes the wavenumber of the probe
laser. For a four-level system, the susceptibility is given by \cite{guo2025aoa}
\begin{small}
    \begin{equation}
\label{eq:susceptibility_full}
\begin{split}
    &\chi(x) ={} \\
    & \frac{j\cfrac{2\pi N_a \mu_{p}^2}{\epsilon_0 \hbar}}{\gamma_{21} - j\Delta_p + \cfrac{\Omega_c^2/4}{\gamma_{31} - j(\Delta_p + \Delta_c) + \cfrac{\Omega_{\text{RF}}^2(x)/4}{\gamma_{41} - j(\Delta_p + \Delta_c + \Delta_{\text{RF}})}}}.
\end{split}
\end{equation}
\end{small}
where $N_a$ is the atomic density, $\mu_{p}$ is the dipole moment of the probe transition, $\epsilon_0$ is the vacuum permittivity, and $\hbar$ is the reduced Planck constant. The terms $\gamma_{21}$, $\gamma_{31}$, and $\gamma_{41}$ are the decay rates for the respective energy levels, while $\Delta_p$, $\Delta_c$, and $\Delta_{\text{RF}}$ are the detunings for the probe, coupling, and RF fields. Under the common experimental conditions of a weak, on-resonance probe, i.e., $\Delta_p=0$, and negligible decay from the Rydberg states, i.e., $\gamma_{31}, \gamma_{41} \ll \gamma_{21}$, this expression simplifies to 
\begin{equation}
\chi(x) \approx j\cfrac{2\pi N_a \mu_{p}^2}{\epsilon_0 \hbar} \cfrac{1}{\gamma_{21} + \cfrac{\Omega_c^2/4}{-j\Delta_c + \cfrac{\Omega_{RF}^2(x)/4}{-j(\Delta_c + \Delta_{\text{RF}})}}}.
\label{eq:susceptibility_simple}
\end{equation}
The response of the system is monitored by collecting the fluorescence emitted from the side of the vapor cell. The fluorescence intensity depends on the local probe laser intensity and the local atomic absorption. For a simple two-level system, the photon scattering rate is given by \cite{metcalf1999laser}
\begin{equation}
\label{eq:scattering_rate}
{R}_{\mathrm{{sc}}} = \frac{\Gamma }{2}\frac{\left( I/{I}_{\mathrm{{sat}}}\right) }{1 + \left( {I/{I}_{\mathrm{{sat}}}}\right)  + 4{\left( \Delta /\Gamma \right) }^{2}},
\end{equation}
where $\Gamma$ is the inverse of the upper state lifetime, $I$ is the laser intensity, ${I}_{\mathrm{{sat}}}$ is the saturation intensity, and $\Delta$ is the laser detuning. In the weak probe regime, i.e.,  $I \ll I_{\text{sat}}$, this scattering rate, and thus the  fluorescence power $P_\text{f}$, is directly proportional to the laser intensity $I$ \cite{schlossberger2024two, radwell20133d}. This relationship provides the basis for our spatially-resolved measurement.

\section{Proposed Approach}\label{III}
The single-receiver method based on integrated power measurement, while innovative, is  limited to single-target scenarios and suffers from a restrictive dependency on a specific, wavelength-dependent cell length. To overcome these limitations, we introduce our Imaging-based Spectral Estimation (ISE) approach. This paradigm shifts from integrated power measurement to spatially-resolved sensing, where by imaging the fluorescence profile along the vapor cell, we can directly access the spatial distribution of the fluorescence power $P_f(x)$.

\subsection{Multi-Signal Linearization under LO-Dominance}
Consider a scenario with $N$ incident plane-wave signals, each with amplitude $A_i$, phase $\phi_i$, and angle of arrival $\theta_i$. When superimposed with an LO field ($A_0, \phi_0, \theta_0$), the total RF field inside the vapor cell is given by
\begin{equation}
E_{\text{RF}}(x) = A_0 e^{j(kx\sin\theta_0+\phi_0)} + \sum_{i=1}^{N} A_i e^{j(kx\sin\theta_i+\phi_i)}.
\label{eq:Erf_multi}
\end{equation}
The intensity of this field, $s(x) = |E_{\text{RF}}(x)|^2$, determines the local atomic absorption coefficient $\alpha(x)$. The full expression for the intensity is complex, as it contains not only the desired signal-LO beat terms but also cross-terms between different signals, which is given by
\begin{equation}
\small 
\begin{split}
&|E_{\text{RF}}(x)|^2 = A_0^2 + \sum_{i=1}^N A_i^2 + 2A_0\sum_{i=1}^N A_i\cos(\Delta k_i x-\Delta\phi_i)\\
&\quad + 2\sum_{1\le i<\ell\le N}A_i A_\ell \cos\big((k\sin\theta_i-k\sin\theta_\ell)x +(\phi_i-\phi_\ell)\big),
\end{split}
\label{eq:multi-square}
\end{equation}
where $\Delta k_i \triangleq k(\sin\theta_0 - \sin\theta_i)$ and $\Delta\phi_i \triangleq \phi_i - \phi_0$. The presence of the signal-signal cross-terms complicates the relationship between the field intensity and the DoAs.

The key to our approach is to operate in an LO-dominant regime, where the LO field is significantly stronger than the sum of the incident signal fields, i.e., $A_0 \gg \sum A_i$. In this regime, the nonlinear relationship between the absorption coefficient $\alpha(x)$ and the field intensity $s(x)$ can be linearized. Combining \eqref{eq:alpha_def} and \eqref{eq:susceptibility_simple}, the absorption coefficient can be written as
\begin{equation}
    \alpha(x) = C \cdot f(s(x)),
\end{equation}
where
\begin{equation}
    f(s) = \cfrac{1}{\gamma_{21}^2+\cfrac{(\Omega_c^2/4)^2}{(\Delta_c-\beta s)^2}},
\end{equation}
and the constants are defined as
\begin{equation}
    C \triangleq \frac{2\pi N_a \mu_p^2 k_{pr} \gamma_{21}}{\epsilon_0 \hbar}, \quad \beta \triangleq \frac{\mu_{RF}^2}{4\hbar^2(\Delta_c+\Delta_{RF})}.
\end{equation}
We perform a first-order Taylor expansion of $\alpha(s)$ around the point $s_0 = A_0^2$ as
\begin{equation}
    \begin{aligned}
        \alpha(x) &\approx \alpha(s_0) + \alpha'(s_0) (s(x) - s_0)\\ &= C f(s_0) + C f'(s_0) (s(x) - s_0),
    \end{aligned}
\label{eq:alpha-linear}
\end{equation}
where the derivative $f'(s)$ is given by
\begin{equation}
f'(s) = -\frac{2\beta (\Omega_c^2/4)^2}{(\Delta_c-\beta s)^3} \cdot \cfrac{1}{\left(\gamma_{21}^2+\cfrac{(\Omega_c^2/4)^2}{(\Delta_c-\beta s)^2}\right)^2}.
\label{eq:fprime}
\end{equation}
The term $s(x) - s_0$ represents the deviation from the LO-only intensity. From  \eqref{eq:multi-square}, this is given by
\begin{align}
    &s(x) - s_0 = \sum_{i=1}^N A_i^2 + 2A_0\sum_{i=1}^N A_i\cos(\Delta k_i x-\Delta\phi_i) \notag \\
    &\quad + 2\sum_{1\le i<\ell\le N}A_i A_\ell \cos\big((k\sin\theta_i-k\sin\theta_\ell)x +(\phi_i-\phi_\ell)\big).
\end{align}
Under the LO-dominant assumption, the signal self-terms and the signal-signal cross-terms  are second-order small quantities and can be neglected. The intensity deviation is thus dominated by the signal-LO beat terms, which yields
\begin{equation}
    s(x) - s_0 \approx 2A_0\sum_{i=1}^N A_i\cos(\Delta k_i x-\Delta\phi_i).
    \label{eq:s_deviation_approx}
\end{equation}
Substituting this approximation into the Taylor expansion in  \eqref{eq:alpha-linear} linearizes the local absorption coefficient into a sum of cosines, yielding
\begin{equation}
\alpha(x) \approx C f(s_0) + C f'(s_0) \left( 2A_0\sum_{i=1}^N A_i\cos(\Delta k_i x-\Delta\phi_i) \right).
\end{equation}
This can be rewritten as
\begin{equation}
\alpha(x) \approx \alpha_{\text{DC}} + \sum_{i=1}^N \mathcal{A}_i \cos(\Delta k_i x - \Delta\phi_i),
\label{eq:alpha-cos}
\end{equation}
where $\alpha_{\text{DC}} = C f(s_0)$ is a constant DC offset and $\mathcal{A}_i$ is the effective amplitude of the spatial modulation for the $i$-th signal, defined as
\begin{equation}
    \mathcal{A}_i \triangleq 2C A_0 A_i f'(s_0).
    \label{eq:Gamma-def}
\end{equation}
Each signal contributes a single cosine term with a spatial frequency $\Delta k_i$ that is uniquely related to its DoA, $\theta_i$. This elegant reduction transforms the complex multi-signal interaction into a simple linear superposition of spatial sinusoids. This finding is significant as it provides a direct pathway to advanced receiver concepts. By treating each resolved spatial frequency as an independent data channel, the system naturally functions as a multi-channel receiver. Furthermore, the spatially-resolved measurement along the vapor cell effectively creates a continuous sensing aperture, a key enabling technology for holographic MIMO architectures.

\subsection{Spatially-Resolved Readout and Spectral Formulation}
Having linearized the atomic response, the next step is to devise a measurement scheme that can extract the individual spatial frequency components $\{\Delta k_i\}$. The process begins with capturing the raw fluorescence profile, $P_f(x)$. 
To relate this captured optical signal to the underlying absorption coefficient $\alpha(x)$, we leverage the Beer-Lambert law, which governs the attenuation of the probe laser power, $P_p(x)$, as it propagates through the vapor: $dP_p(x)/dx = -\alpha(x)P_p(x)$. As noted in Sec.~\ref{II}, in the weak probe regime, the emitted fluorescence power at any point is directly proportional to the probe power at same point, i.e., $P_f(x) \propto P_p(x)$. Combining these two relations yields a direct link between the measurable fluorescence and the absorption coefficient as
\begin{equation}
\alpha(x) = -\frac{d}{dx} \ln P_f(x).
\label{eq:alpha_from_fluorescence}
\end{equation}
This equation shows that the local absorption coefficient can be recovered from the logarithmic derivative of the fluorescence profile, which can be computed numerically from the camera's pixel data.

To convert this continuous spatial information into a discrete data vector suitable for spectral analysis, we apply a set of virtual spatial windows in post-processing. We define a family of $K$ window functions $\{w_j(x)\}_{j=1}^K$, and the measurement for the $j$-th virtual channel, $y_j$, is defined as 
\begin{equation}
    y_j \triangleq \int_0^L w_j(x) \alpha(x) dx.
    \label{eq:yj-def}
\end{equation}
This quantity represents the weighted absorption within the $j$-th window. It is computed from the captured fluorescence data by first calculating $\alpha(x)$ via \eqref{eq:alpha_from_fluorescence} and then performing the weighted integration. For the important special case of a rectangular window where $w_j(x) = 1$ for $x$ within the $j$-th window, e.g., from $x_j$ to $x_{j+1}$, and 0 otherwise, this integral simplifies. By substituting \eqref{eq:alpha_from_fluorescence} and using the fact that $P_f(x) \propto P_p(x)$, the measurement becomes 
\begin{equation}
    y_j = \int_{x_j}^{x_{j+1}} \left(-\frac{d}{dx} \ln P_p(x)\right) dx = -\ln\left(\frac{P_p(x_{j+1})}{P_p(x_j)}\right).
    \label{eq:yj_log_ratio}
\end{equation}
This form, where $P_p(x_j)$ and $P_p(x_{j+1})$ are the input and output powers for the window, provides a practical way to compute the channel measurements.
To eliminate the unknown DC offset  and other constant system parameters, a calibration step is performed. We perform a one-time calibration measurement with the LO field only. In this case, $s(x) = A_0^2 = s_0$, and thus the absorption is uniform, $\alpha_{\text{LO}}(x) = C f(s_0) = \alpha_{\text{DC}}$. The calibration measurement is thus $y_j^{\text{(cal)}} = \int_0^L w_j(x) \alpha_{\text{DC}} dx$. Subtracting this from the signal measurement yields the calibrated signal $\tilde{y}_j = y_j - y_j^{\text{(cal)}}$, which is
\begin{equation}
\tilde{y}_j = \int_0^L w_j(x) (\alpha(x) - \alpha_{\text{DC}}) dx.
\label{eq:ytilde-def}
\end{equation}
Now, we substitute our linearized expression for $\alpha(x)$ from  \eqref{eq:alpha-cos} into this definition. The difference $\alpha(x) - \alpha_{\text{DC}}$ is precisely the sum of spatial cosines
\begin{equation}
\alpha(x) - \alpha_{\text{DC}} \approx \sum_{i=1}^N \mathcal{A}_i \cos(\Delta k_i x - \Delta\phi_i).
\end{equation}
Plugging this into  \eqref{eq:ytilde-def} and exchanging the order of summation and integration yields the final linear mixture model
\begin{small}
\begin{equation}
    \begin{aligned}
\tilde{y}_j &\approx \int_0^L w_j(x) \left( \sum_{i=1}^N \mathcal{A}_i \cos(\Delta k_i x - \Delta\phi_i) \right) dx \notag \\
&= \sum_{i=1}^N \mathcal{A}_i \int_0^L w_j(x) \cos(\Delta k_i x - \Delta\phi_i) dx.
\label{eq:lin-mixture}
\end{aligned} 
\end{equation}
\end{small}
This equation shows that each calibrated measurement $\tilde{y}_j$ is a linear combination of the spatial sinusoids, weighted by the window functions. To create a structure amenable to spectral estimation, we choose the window functions to be a family of shifted versions of a single ``mother" window $w_0(x)$ as
\begin{equation}
w_j(x) = w_0(x - x_j), 
\label{eq:shift-window}
\end{equation}
where $x_j = x_0 + (j-1)\Delta x$. This choice of equispaced window centers effectively creates a virtual sensor array. Compared to the phase-difference measurements of a conventional coherent array, ISE is equivalent to performing spatial heterodyne sampling over a continuous aperture. The window centers $\{x_j\}$ act as the `virtual elements', and the spectral response of the mother window, $\widehat{w}_0(\Delta k_i)$, corresponds to the directional pattern of each element. The integral in \eqref{eq:lin-mixture} is evaluated using the Fourier shift theorem. With the mother window's Fourier transform defined as $\widehat{w}_0(\omega) = \int w_0(x) e^{i\omega x} dx$, we have
\begin{align}
& \int_0^L w_j(x)\cos(\Delta k_i x-\Delta\phi_i)\,dx \notag \\
& \qquad =\Re\!\left[e^{-i\Delta\phi_i}\!\int w_0(x-x_j)e^{i\Delta k_i x}\,dx\right]\notag\\
& \qquad =\Re\!\left[e^{-i\Delta\phi_i}\,e^{i\Delta k_i x_j}\widehat{w}_0(\Delta k_i)\right].
\end{align}
This transforms the measurement model into the canonical form of a sum of real sinusoids, given by
\begin{equation}
\tilde{y}_j = \Re\left[\sum_{i=1}^N b_i e^{i\Delta k_i x_j}\right] = \sum_{i=1}^N |b_i| \cos(\Delta k_i x_j + \arg(b_i)),
\label{eq:sum-of-exps}
\end{equation}
where the complex amplitude $b_i$ is defined as
\begin{equation}
    b_i \triangleq \mathcal{A}_i \widehat{w}_0(\Delta k_i) e^{-i\Delta\phi_i}.
\end{equation}

\begin{theorem}
\label{thm:sampling}
For unambiguous spectral estimation of spatial frequencies, the spatial sampling interval $\Delta x$ and the width $\ell$ of a rectangular window function $w_0(x)$ must satisfy $\Delta x \le \lambda/4$ and $\ell < \lambda/2$, respectively.
\end{theorem}
\begin{proof}
     \noindent\emph{Proof:} See Appendix \ref{Appen.A}.
 \end{proof}
The sequence of measurements $\{\tilde{y}_j\}_{j=1}^K$ now represents a uniformly sampled signal composed of a sum of sinusoids.
\subsection{Analysis of the Integrated-Power Method as a Special Case}
\label{sec:relation_to_guo}
The proposed ISE framework not only enables multi-target detection but also provides a unified model to analyze and clarify existing single-target DoA estimation methods that rely on integrated power measurements \cite{guo2025aoa}. By specializing our model, we can demonstrate that the integrated-power method is a specific instance of our framework and, in doing so, reveal the physical origins of its inherent limitations, particularly regarding the vapor cell length.

An integrated-power method relies on measuring the total power of the probe laser transmitted through the entire vapor cell. This corresponds to a specialization of our virtual channel model, as defined in \eqref{eq:yj-def}, where only a single channel is used with a rectangular window function $w_1(x)$ that spans the entire cell length $L$:
\begin{equation}
    w_1(x) = 1, \quad \text{for } x \in [0, L].
\end{equation}
In this case, the single measurement $y_1$ is the total integrated absorption:
\begin{equation}
    y_1 = \int_0^L \alpha(x) \,dx = -\ln\left(\frac{P_{\rm out}}{P_{\rm in}}\right).
\end{equation}
For a single target, substituting our linearized absorption coefficient from \eqref{eq:alpha-cos} into the integral yields
\begin{align}
    y_1 &= \int_0^L \left[ \alpha_{\text{DC}} + \mathcal{A}_1 \cos(\Delta k_1 x - \Delta\phi_1) \right] dx \nonumber \\
    &= \alpha_{\text{DC}}L + \mathcal{A}_1 \int_0^L \cos(\Delta k_1 x - \Delta\phi_1) dx.
    \label{eq:integrated_absorption}
\end{align}
The total power transmission $T(\theta_1) = \exp(-y_1)$ is then used to estimate $\theta_1$. This derivation explicitly shows that the model in \cite{guo2025aoa} is mathematically equivalent to a single-channel, single-target specialization of the ISE framework under the strong LO approximation.

This unified model also allows us to analytically explain the empirically determined optimal cell length in such methods. The estimation relies on inverting the mapping from the DoA $\theta_1$ to the transmission $T(\theta_1)$. To avoid ambiguity and ensure a well-behaved, convex optimization landscape, this mapping should be monotonic. The core of this mapping is the integral term in \eqref{eq:integrated_absorption}, which, for $\Delta\phi_1 \approx 0$, is
\begin{equation}
    \int_0^L \cos(\Delta k_1 x) dx = L\,\frac{\sin(\Delta k_1 L)}{\Delta k_1 L} = L\,\mathrm{sinc}(\Delta k_1 L).
\end{equation}
This integral is the Fourier transform of the rectangular window function, which has a sinc-like shape. To maintain monotonicity, the argument of the sinc function must be confined to its main lobe. The first zero of the derivative of $\mathrm{sinc}(u)$ occurs at the first non-zero root of the transcendental equation $u = \tan(u)$, which is $u_1 \approx 4.493$. To ensure a monotonic response over the full range of possible angles, the maximum argument must be less than this root: $|\Delta k_1 L| < u_1$.
The spatial frequency range is bounded by $|\Delta k_1| \le 2k = 4\pi/\lambda_{\text{RF}}$. Therefore, the condition becomes
\begin{equation}
    (2k)L < u_1 \implies L < \frac{u_1}{2k} = \frac{4.493}{4\pi/\lambda_{\text{RF}}} \approx 0.358\lambda_{\text{RF}}.
\end{equation}
This result provides a clear analytical explanation for the empirically found optimum in \cite[Fig. 9-11]{guo2025aoa}, which is $L \approx \lambda_{\text{RF}}/3$. This cell length represents a trade-off: it must be long enough to resolve the interference pattern but short enough to keep the response monotonic.

However, this constraint is not inherent to Rydberg sensors in general. It forces a compromise between sensitivity and unambiguous range, and restricts the broadband potential of the sensor, as the optimal physical length $L$ changes with the wavelength $\lambda_{\text{RF}}$. In contrast, our ISE approach, by spatially resolving the interference pattern, is not bound by this monotonicity constraint. It effectively performs spectral analysis on the spatial signal, where a longer cell length $L$ corresponds to a larger sensing aperture, leading to improved spectral resolution and thus higher DoA accuracy. This key difference liberates the sensor from the restrictive cell-length dependency and restores its full broadband, high-resolution capabilities.
\subsection{Spectral Estimation via Prony's Method}
\label{sec:prony}
\begin{figure*}[htbp]
\centering
\includegraphics[width=0.8\textwidth]{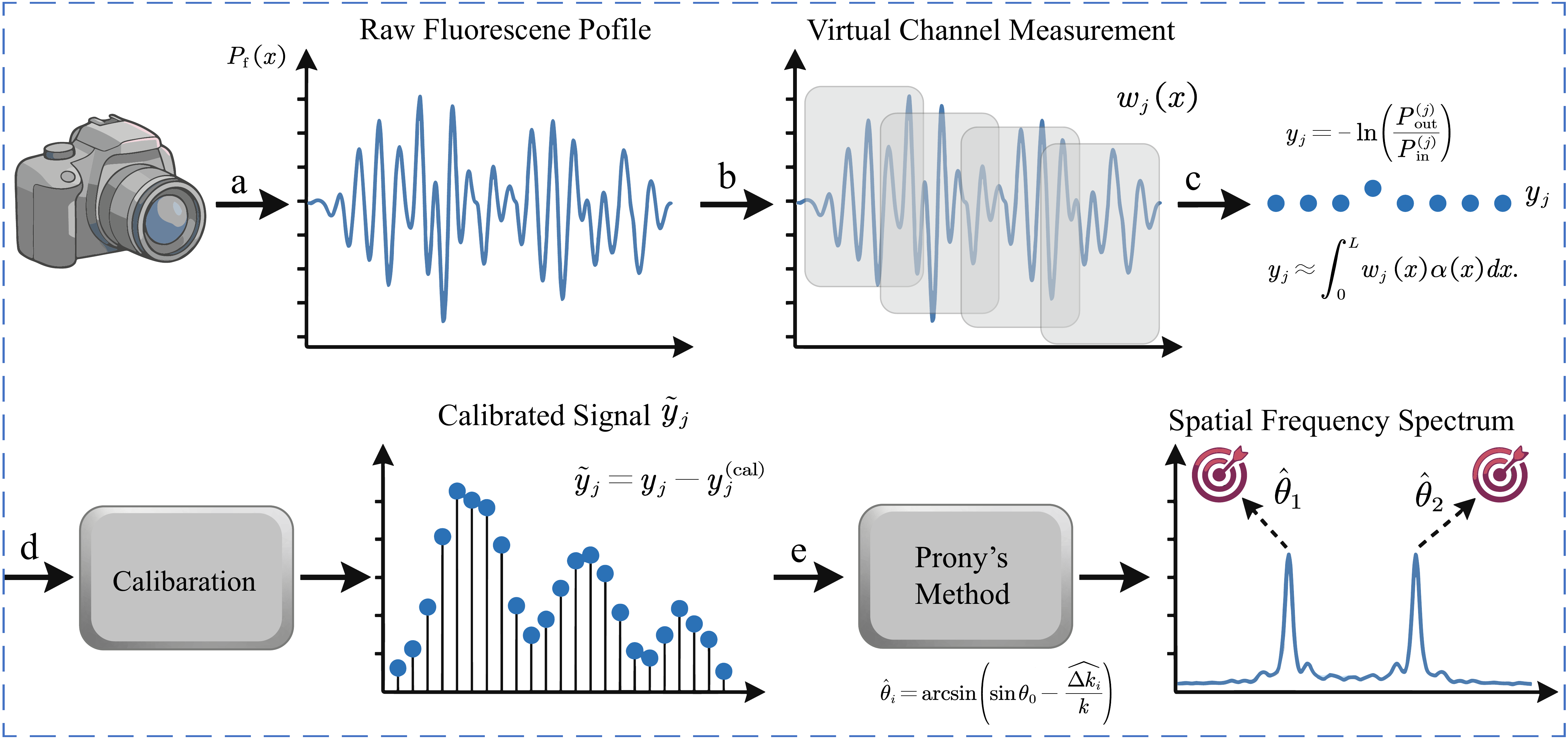}
\caption{Flow diagram of the proposed ISE approach. The process begins with (a) capturing the raw fluorescence profile. (b) A set of spatial windows are applied, and (c) the integrated absorption within each window yields the virtual channel measurements $y_j$. (d) A calibration step removes background effects, producing the calibrated signal $\tilde{y}_j$. (e) Prony's method is applied to the calibrated signal to estimate the spatial frequencies, which are then mapped to the final DoA estimates.}
\label{fig:flow_diagram}
\end{figure*}
 According to \eqref{eq:sum-of-exps}, the problem of estimating the DoAs $\{\theta_i\}$ has been successfully reduced to the classic problem of estimating the frequencies $\{\Delta k_i\}$ from the calibrated measurement sequence. The final step is to apply Prony's method to solve this spectral estimation problem. While traditionally sensitive to noise, its application is particularly well-suited to our context. Rydberg atomic sensors are fundamentally immune to the thermal noise that plagues conventional receivers, mitigating the primary drawback of the method. Furthermore, compared to other high-resolution algorithms that often require multiple snapshots to build a statistical covariance matrix, Prony's method offers the distinct advantages of having very low computational complexity and requiring only a single snapshot of the spatial data. This makes it highly efficient for this application. Once the spatial frequencies $\{\widehat{\Delta k}_i\}$ are estimated, the corresponding DoAs are recovered via the mapping in  \eqref{eq:field_intensity}.
 
Prony's method is a well-known technique for estimating the parameters of a sum of sinusoids from a series of samples. We apply it here to solve the estimation problem that arises from our physical model. As established in \eqref{eq:sum-of-exps}, the sequence of calibrated measurements $\{\tilde{y}_j\}$ for $j=1, \cdots, K$ is a sum of sinusoids with spatial sampling points $x_j = x_0 + (j-1)\Delta x$. The core principle of Prony's method is that such a signal is the solution to a linear constant-coefficient difference equation. This allows for the direct estimation of the frequencies $\{\Delta k_i\}$ from the $K$ samples.

First, the signal sequence $\tilde{y}_j$ is assumed to satisfy a $p$-th order linear prediction model, where $p \ge 2N$, as
\begin{equation}
\tilde y_j = -\sum_{m=1}^p a_m \tilde y_{j-m}, \quad j=p+1, \cdots, K.
\label{eq:prony-lpc}
\end{equation}
The prediction coefficients $\bm{a} = [a_1, \cdots, a_p]^\top$ are found by solving an overdetermined system of linear equations. This system is constructed by forming a Hankel matrix from the measurements, which is given by
\begin{equation}
\begin{bmatrix}
\tilde y_p & \tilde y_{p-1} & \cdots & \tilde y_1 \\
\tilde y_{p+1} & \tilde y_p & \cdots & \tilde y_2 \\
\vdots & \vdots & \ddots & \vdots \\
\tilde y_{K-1} & \tilde y_{K-2} & \cdots & \tilde y_{K-p}
\end{bmatrix}
\begin{bmatrix}
a_1 \\ a_2 \\ \vdots \\ a_p
\end{bmatrix}
=
-
\begin{bmatrix}
\tilde y_{p+1} \\ \tilde y_{p+2} \\ \vdots \\ \tilde y_K
\end{bmatrix}.
\end{equation}
This system is solved for $\bm{a}$ in a least-squares sense. With the coefficients determined, we form the characteristic polynomial given by
\begin{equation}
P(z) = 1 + \sum_{m=1}^p a_m z^{-m}.
\label{eq:prony-poly}
\end{equation}
The roots of this polynomial, $\{z_m\}_{m=1}^p$, contain the frequency information. Each spatial frequency $\Delta k_i$ corresponds to a pair of complex conjugate roots, $z_i$ and $z_i^*$, on the unit circle, such that $z_i = e^{j \Delta k_i \Delta x}$. The frequencies are therefore estimated by calculating the angle of the roots as
\begin{equation}
\widehat{\Delta k_i} = \frac{\arg(z_i)}{\Delta x}.
\end{equation}
We retain the unique positive frequencies. In the presence of noise, the roots may deviate from the unit circle; the signal roots are identified as those with magnitudes closest to 1, distinguishing them from spurious roots introduced by noise.

Once the set of spatial frequencies $\{\widehat{\Delta k_i}\}_{i=1}^N$ is estimated, the corresponding DoAs are recovered via the inverse mapping as
\begin{equation}
    \hat{\theta}_i = \arcsin\left(\sin\theta_0 - \frac{\widehat{\Delta k_i}}{k}\right).
\end{equation}
This final step recovers the desired directional information for all incident signals, completing the estimation process.

The computational complexity of Prony's method is dominated by two main steps. First, solving the overdetermined system of linear equations in \eqref{eq:prony-lpc} to find the linear prediction coefficients $\{a_m\}$. Using a standard least-squares solver, this step has a complexity of approximately $\mathcal{O}(p^2 K)$, where $p$ is the chosen model order and $K$ is the number of spatial samples. Second, finding the roots of the $p$-th degree characteristic polynomial in \eqref{eq:prony-poly}, which typically has a complexity of $\mathcal{O}(p^2)$. Since in practice $K > p$, the overall complexity is dominated by the first step and is approximately $\mathcal{O}(p^2 K)$. This is significantly more efficient than many other spectral estimation algorithms, especially those requiring covariance matrix estimation, making it well-suited for real-time applications.

The entire process, from fluorescence imaging to DoA recovery, is summarized in Fig.~\ref{fig:flow_diagram}.

\section{Performance Analysis: Cramér-Rao Lower Bound}\label{sec:crlb}
To establish a fundamental benchmark for the performance of our proposed DoA estimation approach, we derive the CRLB. The CRLB provides the minimum achievable variance for  unbiased estimators.

\subsection{Statistical Model and Fisher Information Matrix}
Our estimation problem is based on the calibrated measurement vector $\tilde{\bm{y}} = [\tilde{y}_1, \cdots, \tilde{y}_K]^\top$. We model the measurements as the true signal plus additive noise, expressed as
\begin{equation}
    \tilde{\bm{y}} = \bm{\mu}(\bm{\vartheta}) + \bm{n},
\end{equation}
where $\bm{\mu}(\bm{\vartheta}) = \mathbb{E}[\tilde{\bm{y}}]$ is the mean vector from \eqref{eq:lin-mixture}, $\bm{n}$ is a zero-mean Gaussian noise vector with covariance matrix $\bm{\Sigma}_y$ that encapsulates various noise sources (e.g., photon shot noise, laser noise, and electronic noise), and $\bm{\vartheta}$ is the vector of all unknown parameters for the $N$ targets, given by
\begin{equation}
    \bm{\vartheta} = [\bm{\Delta k}^\top, \bm{\Delta\phi}^\top, \bm{\mathcal{A}}^\top]^\top \in \mathbb{R}^{3N},
\end{equation}
with $\bm{\Delta k} = [\Delta k_1, \cdots, \Delta k_N]^\top$, $\bm{\Delta\phi} = [\Delta\phi_1, \cdots, \Delta\phi_N]^\top$, and $\bm{\mathcal{A}} = [\mathcal{A}_1, \cdots, \mathcal{A}_N]^\top$.
For a Gaussian model where the covariance is independent of the parameters, the Fisher Information Matrix (FIM) is given by \cite{kay1993fundamentals}
\begin{equation}
    \bm{I}(\bm{\vartheta}) = \left(\frac{\partial \bm{\mu}(\bm{\vartheta})}{\partial \bm{\vartheta}}\right)^\top \bm{\Sigma}_y^{-1} \left(\frac{\partial \bm{\mu}(\bm{\vartheta})}{\partial \bm{\vartheta}}\right).
    \label{eq:FIM-general}
\end{equation}
The FIM quantifies the amount of information the measurement vector $\tilde{\bm{y}}$ provides about the unknown parameters $\bm{\vartheta}$, and its inverse sets the lower bound on the variance of any unbiased estimator. The Jacobian matrix $\frac{\partial \bm{\mu}}{\partial \bm{\vartheta}}$ can be constructed from the partial derivatives with respect to each parameter. Let $\bm{c}_i$, $\bm{s}_i$, and $\bm{t}_i$ be $K \times 1$ vectors whose $j$-th elements are
\begin{small}
    \begin{equation}
\begin{aligned}
    [\bm{c}_i]_j &= \int_0^L w_j(x)\cos(\Delta k_i x-\Delta\phi_i)\,dx, \\
    [\bm{s}_i]_j &= \int_0^L w_j(x)\sin(\Delta k_i x-\Delta\phi_i)\,dx, \\
    [\bm{t}_i]_j &= -\int_0^L w_j(x) x \sin(\Delta k_i x-\Delta\phi_i)\,dx.
\end{aligned}
\end{equation}
\end{small}
The partial derivatives of the mean vector are then given by
\begin{small}
    \begin{equation}
    \begin{aligned}
    \frac{\partial \bm{\mu}}{\partial \mathcal{A}_i} &= \bm{c}_i, \\
    \frac{\partial \bm{\mu}}{\partial \Delta\phi_i} &= \mathcal{A}_i \bm{s}_i, \\
    \frac{\partial \bm{\mu}}{\partial \Delta k_i} &= \mathcal{A}_i \bm{t}_i.
\end{aligned}
\end{equation}
\end{small}
Using these, the FIM can be expressed in a $3 \times 3$ block matrix form, where each block is an $N \times N$ matrix. Defining the inner product as $\langle\bm{a},\bm{b}\rangle_{\bm{\Sigma}^{-1}} \triangleq \bm{a}^\top\bm{\Sigma}_y^{-1}\bm{b}$, the blocks of the FIM are given by 
\begin{equation}
\begin{split}
    [\bm{I}_{\bm{\mathcal{A}}\bm{\mathcal{A}}}]_{i\ell} &= \langle\bm{c}_i,\bm{c}_\ell\rangle_{\bm{\Sigma}^{-1}}, \\
    [\bm{I}_{\bm{\Delta\phi}\bm{\Delta\phi}}]_{i\ell} &= \mathcal{A}_i\mathcal{A}_\ell\,\langle\bm{s}_i,\bm{s}_\ell\rangle_{\bm{\Sigma}^{-1}}, \\
    [\bm{I}_{\bm{\Delta k}\bm{\Delta k}}]_{i\ell} &= \mathcal{A}_i\mathcal{A}_\ell\,\langle\bm{t}_i,\bm{t}_\ell\rangle_{\bm{\Sigma}^{-1}}, \\
    [\bm{I}_{\bm{\mathcal{A}}\bm{\Delta\phi}}]_{i\ell} &= \mathcal{A}_\ell\,\langle\bm{c}_i,\bm{s}_\ell\rangle_{\bm{\Sigma}^{-1}}, \\
    [\bm{I}_{\bm{\mathcal{A}}\bm{\Delta k}}]_{i\ell} &= \mathcal{A}_\ell\,\langle\bm{c}_i,\bm{t}_\ell\rangle_{\bm{\Sigma}^{-1}}, \\
    [\bm{I}_{\bm{\Delta\phi}\bm{\Delta k}}]_{i\ell} &= \mathcal{A}_i\mathcal{A}_\ell\,\langle\bm{s}_i,\bm{t}_\ell\rangle_{\bm{\Sigma}^{-1}}.
\end{split}
\end{equation}
where $i, \ell = 1, \cdots, N$.
The block structure of the FIM provides insight into the estimation problem. The diagonal blocks, e.g., $\bm{I}_{\bm{\Delta k}\bm{\Delta k}}$, represent the information available for estimating a specific set of parameters. The off-diagonal blocks, e.g., $\bm{I}_{\bm{\Delta k}\bm{\Delta\phi}}$, quantify the statistical coupling between different sets of parameters. A large off-diagonal coupling implies that the estimation of one parameter set is highly correlated with the estimation of another, making it more challenging to distinguish their individual effects. For instance, if two targets are closely spaced in angle, their corresponding columns in the Jacobian matrix become nearly collinear, resulting in an ill-conditioned FIM and a higher variance for their estimates, which aligns with the physical intuition of reduced resolution.
The CRLB for the entire parameter vector $\bm{\vartheta}$ is given by the inverse of the FIM, i.e., $\mathrm{Cov}(\hat{\bm{\vartheta}}) \succeq \bm{I}(\bm{\vartheta})^{-1}$.
\vspace{-0.3cm}
\subsection{CRLB for Angle of Arrival}
In our problem, the parameters of interest are the spatial frequencies $\bm{\Delta k}$, which directly map to the DoAs $\bm{\theta}$. The amplitudes $\bm{\mathcal{A}}$ and phases $\bm{\Delta\phi}$ are considered nuisance parameters. The CRLB for $\bm{\Delta k}$ is found by first computing the effective FIM for $\bm{\Delta k}$, which is obtained by inverting the full FIM $\bm{I}(\bm{\vartheta})$ and taking the sub-block corresponding to $\bm{\Delta k}$. This is equivalent to using the Schur complement, which is
\begin{equation}
    \bm{I}_{\text{eff}}^{(\bm{\Delta k})} = \bm{I}_{\bm{\Delta k}\bm{\Delta k}} - \bm{I}_{\bm{\Delta k}\bm{\eta}} \bm{I}_{\bm{\eta}\bm{\eta}}^{-1} \bm{I}_{\bm{\eta}\bm{\Delta k}},
    \label{eq:EFIM-general}
\end{equation}
where $\bm{\eta} = [\bm{\Delta\phi}^\top, \bm{\mathcal{A}}^\top]^\top$ is the vector of nuisance parameters. Physically, the Schur complement accounts for the performance degradation that arises from having to estimate the nuisance parameters $\bm{\eta}$ simultaneously with the parameters of interest $\bm{\Delta k}$. The term subtracted from $\bm{I}_{\bm{\Delta k}\bm{\Delta k}}$ represents the information about $\bm{\Delta k}$ that is lost because of uncertainty in $\bm{\eta}$. The CRLB for any unbiased estimator $\widehat{\bm{\Delta k}}$ is then given by
\begin{equation}
    \mathrm{Cov}(\widehat{\bm{\Delta k}}) \succeq \left(\bm{I}_{\text{eff}}^{(\bm{\Delta k})}\right)^{-1}.
\end{equation}
Finally, the CRLB for the DoA vector $\bm{\theta}$ is obtained by a change of variables. The Fisher information matrix for $\bm{\theta}$ is related to that of $\bm{\Delta k}$ via the Jacobian of the parameter transformation, where
\begin{equation}
    \bm{I}_{\text{eff}}^{(\bm{\theta})} = \bm{J}^\top \bm{I}_{\text{eff}}^{(\bm{\Delta k})} \bm{J},
\end{equation}
where $\bm{J}$ is the Jacobian matrix of the transformation from $\bm{\theta}$ to $\bm{\Delta k}$. From the relation $\Delta k_i = k(\sin\theta_0 - \sin\theta_i)$, this Jacobian is a diagonal matrix given by
\begin{equation}
    \bm{J} = \frac{\partial \bm{\Delta k}}{\partial \bm{\theta}} = \mathrm{diag}(-k\cos\theta_1, \cdots, -k\cos\theta_N).
\end{equation}
The CRLB for $\bm{\theta}$ is the inverse of its effective FIM, which is given by
\begin{equation}
    \mathrm{Cov}(\widehat{\bm{\theta}}) \succeq \left(\bm{I}_{\text{eff}}^{(\bm{\theta})}\right)^{-1} = \bm{J}^{-1} \left(\bm{I}_{\text{eff}}^{(\bm{\Delta k})}\right)^{-1} (\bm{J}^{-1})^\top.
    \label{eq:CRLB-theta-matrix}
\end{equation}
For the single-target case ($N=1$), this simplifies to
\begin{equation}
    \mathrm{var}(\hat{\theta}) \ge \frac{1}{k^2\cos^2\theta} \cdot \frac{1}{I_{\text{eff}}^{(\Delta k)}}.
    \label{eq:CRLB-theta-single}
\end{equation}
 This expression provides a clear physical interpretation of the fundamental limits of the estimation. The term $1/(k^2\cos^2\theta)$ is a purely geometric factor. Since $k=2\pi/\lambda$, the $k^2$ term shows that the estimation precision improves quadratically with decreasing wavelength, as expected. The $\cos^2\theta$ term reveals the strong dependence on the DoA itself. The precision is highest for signals arriving at broadside, i.e., $\theta=0$, and degrades significantly as the angle approaches end-fire i.e., $\theta \to \pm 90^\circ$, where the bound diverges. This is because the projected spatial frequency $\Delta k_i = k(\sin\theta_0 - \sin\theta_i)$ becomes insensitive to changes in $\theta_i$ near the ends of the cell. The second term, $1/I_{\text{eff}}^{(\Delta k)}$, encapsulates the performance contribution of the physical system itself, including the signal-to-noise ratio via $\mathcal{A}_i$ and $\bm{\Sigma}_y$, the number of virtual sensor elements $K$, and the spatial sampling strategy through the window functions and their spacing.

\section{Simulation Results}\label{IV}
To validate the theoretical framework, we conduct a series of numerical simulations. This section first details the physical parameters, which are consistent with established experimental setups \cite{guo2025aoa, holloway2014broadband}, and then presents the results that visually and quantitatively support our proposed approach.

\subsection{Simulation Parameters}
The parameters used in our simulations are based on a realistic four-level Rb atomic system and are consistent with those reported in prior experimental and theoretical works \cite{guo2025aoa, holloway2014broadband}. The key parameters are summarized in Table \ref{tab:sim_params}.

\begin{table}[h]
\centering
\caption{Simulation Parameters}
\label{tab:sim_params}
\begin{tabular}{@{}lll@{}}
\toprule
\textbf{Parameter} & \textbf{Symbol} & \textbf{Value} \\ \midrule
\multicolumn{3}{l}{\textbf{Atomic and Laser Properties}} \\
Atom Density & $N_a$ & $4.13 \times 10^{13}$ m$^{-3}$ \\
Probe Wavelength & $\lambda_{pr}$ & 780.24 nm \\
Coupling Wavelength & $\lambda_{cp}$ & 480 nm \\
Probe Dipole Moment & $\mu_{p}$ & $1.06 \times 10^{-29}$ C$\cdot$m \\
RF Dipole Moment & $\mu_{RF}$ & $7.85 \times 10^{-26}$ C$\cdot$m \\
Decay Rate & $\gamma_{21}$ & $2\pi \times 6.066$ MHz \\
Coupling Rabi Freq. & $\Omega_c$ & $2\pi \times 40$ MHz \\
Coupling Detuning & $\Delta_c$ & $2\pi \times 10$ kHz \\
\midrule
\multicolumn{3}{l}{\textbf{RF Signal Properties}} \\
RF Carrier Frequency & $f_c$ & 2.03 GHz \\
LO Angle of Arrival & $\theta_0$ & $90^\circ$ \\
\bottomrule
\end{tabular}
\end{table}
\begin{figure}[h]
\centering
\includegraphics[width=\columnwidth]{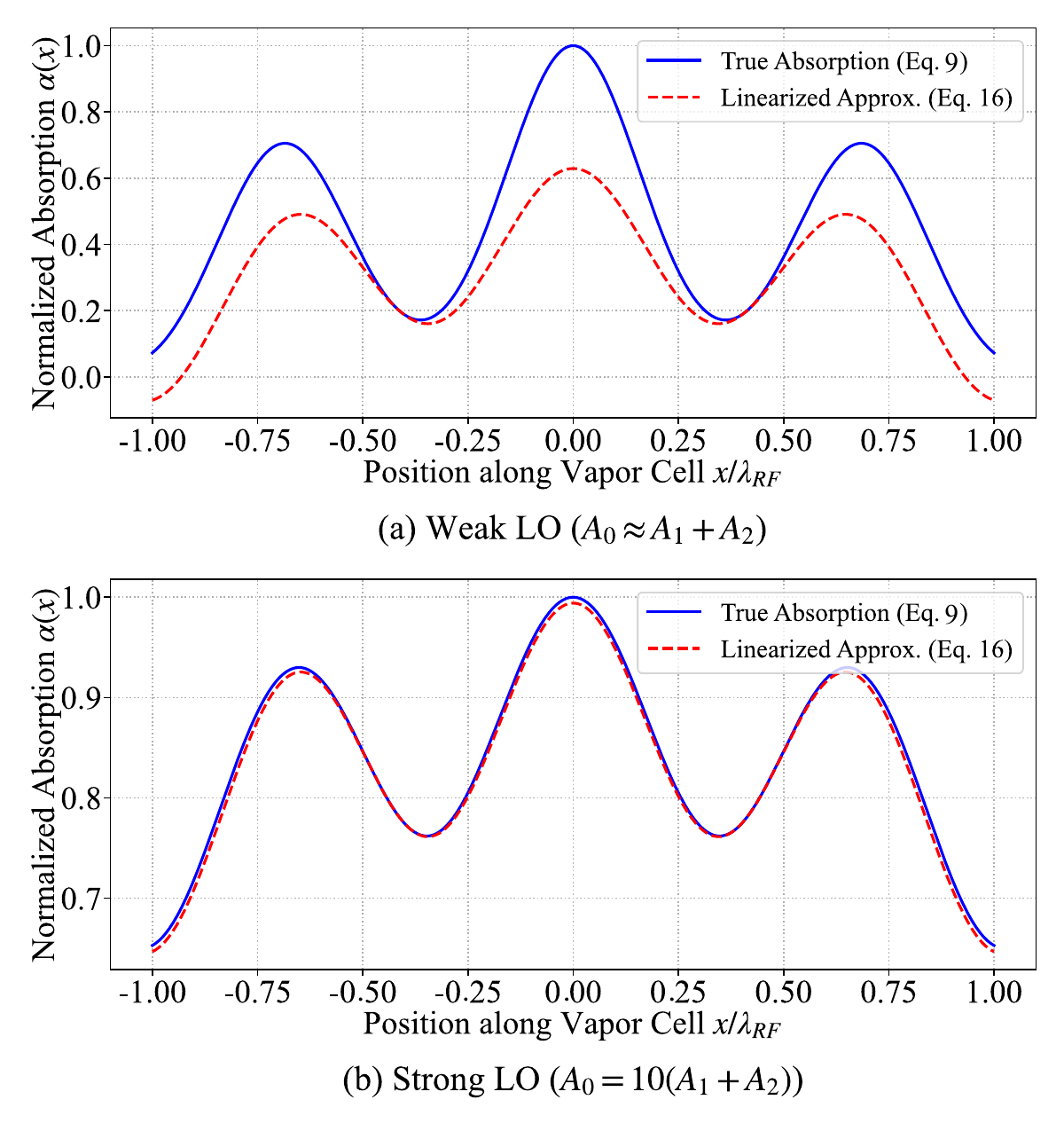}
  \caption{Validation of the linearization assumption. The true absorption coefficient is compared to the linearized model under (a) weak LO and (b) strong LO conditions.}
  \label{fig:linearization_validation}
\end{figure}
\subsection{Validation of the Linearization Assumption}
A foundation of our proposed approach is the linearization of the atomic absorption coefficient under a LO-dominant regime, as formulated in \eqref{eq:alpha-cos}. To  verify this critical assumption, we simulate a two-target scenario and compare the true, non-linear absorption profile with its first-order linear approximation.

Fig. \ref{fig:linearization_validation} illustrates this comparison under two distinct conditions: a   ``weak LO" scenario where the LO amplitude is comparable to the sum of the target signal amplitudes, i.e., $A_0 \approx A_1 + A_2$, and a ``strong LO" scenario where the LO is an order of magnitude stronger, e.g., $A_0 = 10(A_1 + A_2)$.

Fig. \ref{fig:linearization_validation} illustrates this comparison under two distinct conditions. In the weak LO case, where the LO amplitude is comparable to the sum of the target signal amplitudes, the true absorption curve deviates significantly from the linearized model. The waveform is distorted, reflecting the influence of non-negligible signal-signal cross-terms and other higher-order nonlinearities that are not captured by the first-order Taylor expansion.

Conversely, in the strong LO case, where the LO is an order of magnitude stronger, the agreement between the true absorption and the linearized model is striking, as shown in Fig. \ref{fig:linearization_validation}(b). The solid blue line almost perfectly overlays the dashed red line. This demonstrates that when the LO field is dominant, the cross-terms and higher-order effects are effectively suppressed, and the local absorption coefficient $\alpha(x)$ indeed behaves as a simple superposition of a DC offset and several cosine functions. This result provides strong validation for our central hypothesis: under LO-dominance, the multi-target DoA problem can be elegantly and accurately transformed into a spectral estimation problem.

\subsection{Multi-Target Resolution and Impact of LO Strength}

\begin{figure}[t]
\centering
\includegraphics[width=\columnwidth]{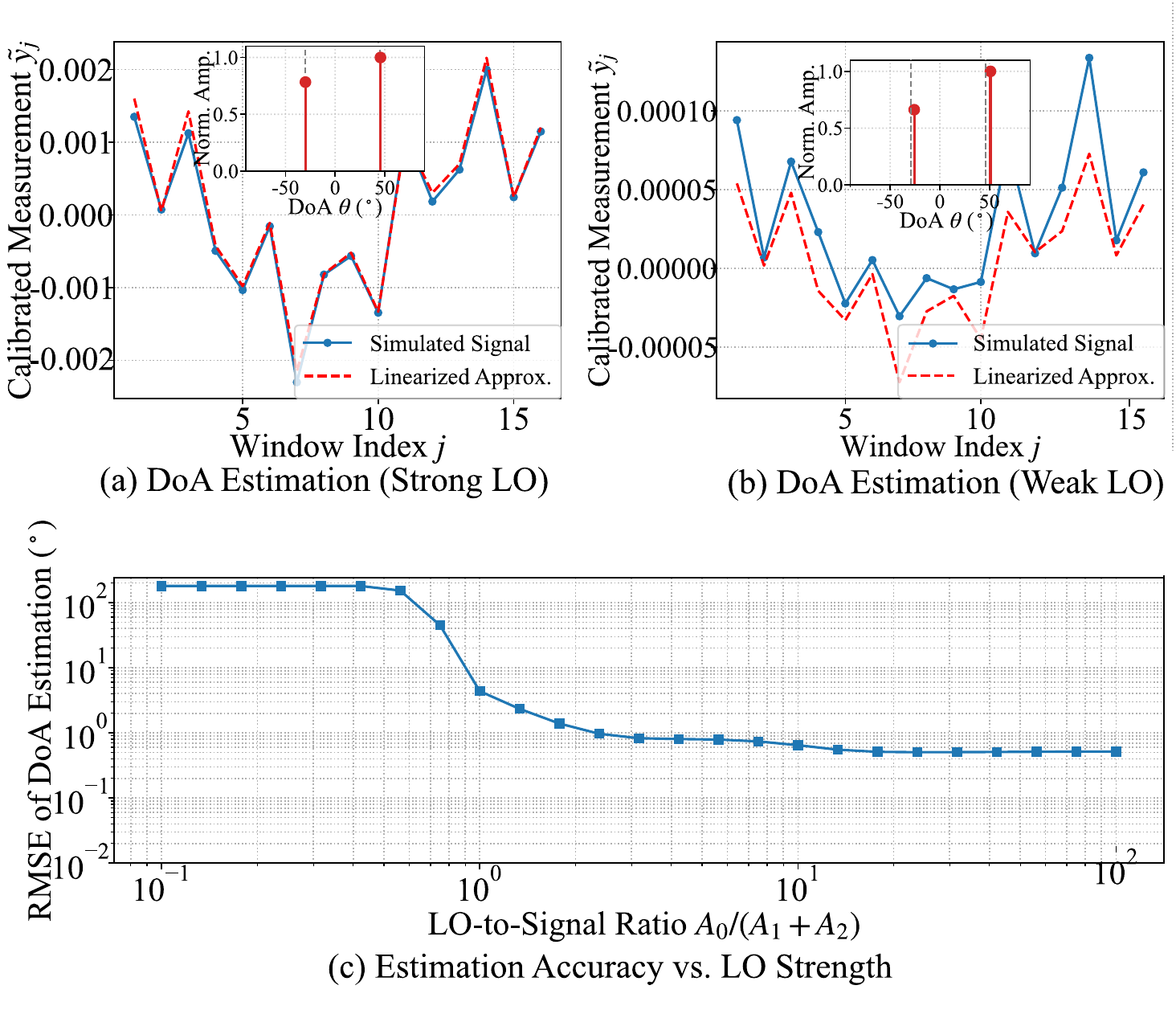}
\caption{Demonstration of the spectral estimation approach and the impact of LO strength. (a) Accurate DoA estimation under strong LO. (b) Degraded estimation under weak LO. (c) RMSE of DoA estimation versus LO-to-signal ratio.}
\label{fig:ise_prony_demo}
\end{figure}

Fig.~\ref{fig:ise_prony_demo} presents a comprehensive demonstration of the proposed approach's capability to resolve multiple targets and highlights the critical role of the LO strength. The simulation involves two incoming signals with true DoAs at $\theta_1 = -30^\circ$ and $\theta_2 = 45^\circ$, and an SNR of 30~dB. For these simulations, the vapor cell length is set to $L=4\lambda_{\text{RF}}$, the fluorescence profile is sampled into $K=16$ virtual channels, and the LO signal is directed at $\theta_0=90^\circ$.

Fig.~\ref{fig:ise_prony_demo}(a) shows the result under a strong LO condition, where the LO-to-signal amplitude ratio is 20. The main plot displays the calibrated spatial measurements, $\tilde{y}_j$. The simulated signal aligns almost perfectly with the prediction from our linearized model, confirming the validity of the approximation in this regime. The inset provides the corresponding DoA estimation from Prony's method, where the estimated angles show excellent agreement with the true DoAs which is presented by gray dashed lines, validating the method's accuracy.

In contrast, Fig.~\ref{fig:ise_prony_demo}(b) illustrates the outcome under a weak LO condition, with a ratio of 1. In this scenario, the linearization assumption fails, causing the spatial measurement waveform to be visibly distorted and deviate  from the sum-of-sinusoids model. This discrepancy is a clear manifestation of the non-linear effects that our theory aims to suppress. As a direct consequence, the accuracy of the DoA estimates shown in the inset is degraded, resulting in a larger estimation error.

Fig.~\ref{fig:ise_prony_demo}(c) provides the quantitative underpinning for these observations. It plots the Root-Mean-Square Error (RMSE) of the DoA estimation as a function of the LO-to-signal ratio, based on 100 Monte Carlo trials per point. The curve shows that for ratios below 1, the error is extremely high, indicating a failure of the estimation. As the ratio increases, the RMSE drops dramatically, eventually reaching a floor determined by the system's SNR for ratios greater than approximately 10. This result not only explains the success of the estimation in Fig.~\ref{fig:ise_prony_demo}(a) and the degradation in Fig.~\ref{fig:ise_prony_demo}(b) but also provides a crucial design guideline for practical implementation: to ensure high-accuracy DoA estimation, the LO power must be at least an order of magnitude greater than the expected total signal power.
\subsection{RMSE Performance vs. SNR}
\begin{figure}[htbp]
\centering
\includegraphics[width=\columnwidth]{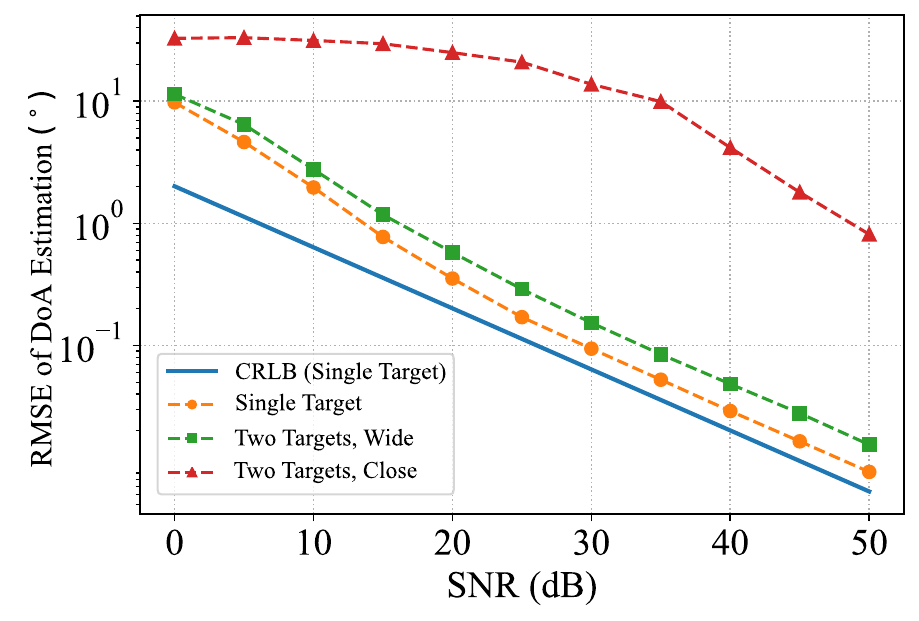}
\caption{RMSE of DoA estimation versus SNR for single and multiple targets. The performance of Prony's method is compared against the CRLB for the single-target case.}
\label{fig:rmse_vs_snr}
\end{figure}
Fig.~\ref{fig:rmse_vs_snr} evaluates the performance of the proposed approach by plotting the RMSE of the DoA estimation as a function of SNR. The simulation compares three distinct scenarios: a single target at $15^\circ$, two widely separated targets at $-15^\circ$ and $15^\circ$, and two closely separated targets at $15^\circ$ and $20^\circ$. For reference, the theoretical CRLB for the single-target case is also plotted, representing the fundamental lower bound on estimation variance.

The results clearly demonstrate that for the single-target case, the RMSE of the proposed approach closely follows the CRLB, especially at higher SNRs, confirming the near-optimal performance of our approach. As expected, the estimation for two targets requires a higher SNR to achieve the same level of accuracy as the single-target case. Furthermore, the RMSE for the two closely-spaced targets is consistently higher than for the widely-spaced targets, illustrating the increased difficulty in resolving targets as their angular separation decreases. These findings are in great agreement with the principles of spectral estimation and validate the robustness and effectiveness of the proposed approach for multi-target DoA estimation.

\subsection{Impact of Aperture Size and Broadband Performance}
\begin{figure}[h]
\centering
\includegraphics[width=\columnwidth]{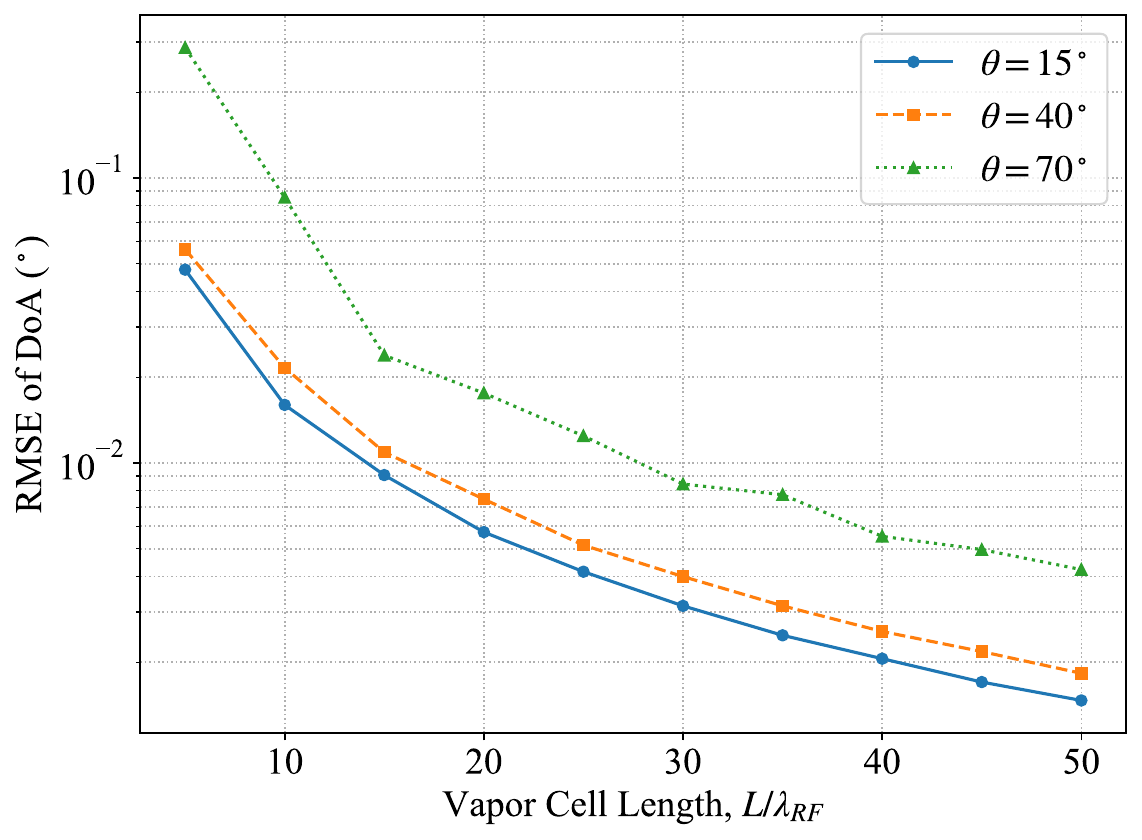}
\caption{RMSE of DoA Estimation vs. Normalized Vapor Cell Length.}
\label{fig:crlb_vs_length}
\end{figure}
Fig.~\ref{fig:crlb_vs_length} illustrates the relationship between the DoA estimation accuracy and the physical size of the sensing medium, a key aspect of our proposed method. The simulation calculates the CRLB for the RMSE of DoA estimation as a function of the vapor cell length, normalized to the RF wavelength. The simulation was performed for three different true signal angles with a fixed SNR of 30~dB, a LO angle of $90^\circ$, and a spatial sampling interval of $\Delta x = \lambda/4$.

The results show a clear and consistent trend: for all angles, the RMSE steadily decreases as the normalized cell length increases. This behavior is attributed to the increase in the effective sensing aperture. A longer cell allows for the collection of more spatial cycles of the interference pattern, providing a richer dataset for spectral estimation and thus leading to a more precise DoA estimate. This result stands in contrast to previous single-receiver methods that are  sensitive to a specific, optimal cell length. Our approach, by leveraging spatially-resolved measurements, removes this restrictive dependency and demonstrates robust performance improvement with increasing aperture size. Crucially, this highlights the inherent broadband capability of our method. For a fixed physical cell length $L$, higher frequency signals will correspond to a larger electrical length $L/\lambda_{\text{RF}}$. As the figure shows, a larger electrical length directly translates to improved estimation accuracy. This means a single, fixed-length Rydberg sensor can achieve progressively higher precision when operating at higher frequencies, a significant advantage for broadband applications.

\subsection{Validation of Sampling Constraints}
\begin{figure}[h]
\centering
\includegraphics[width=\columnwidth]{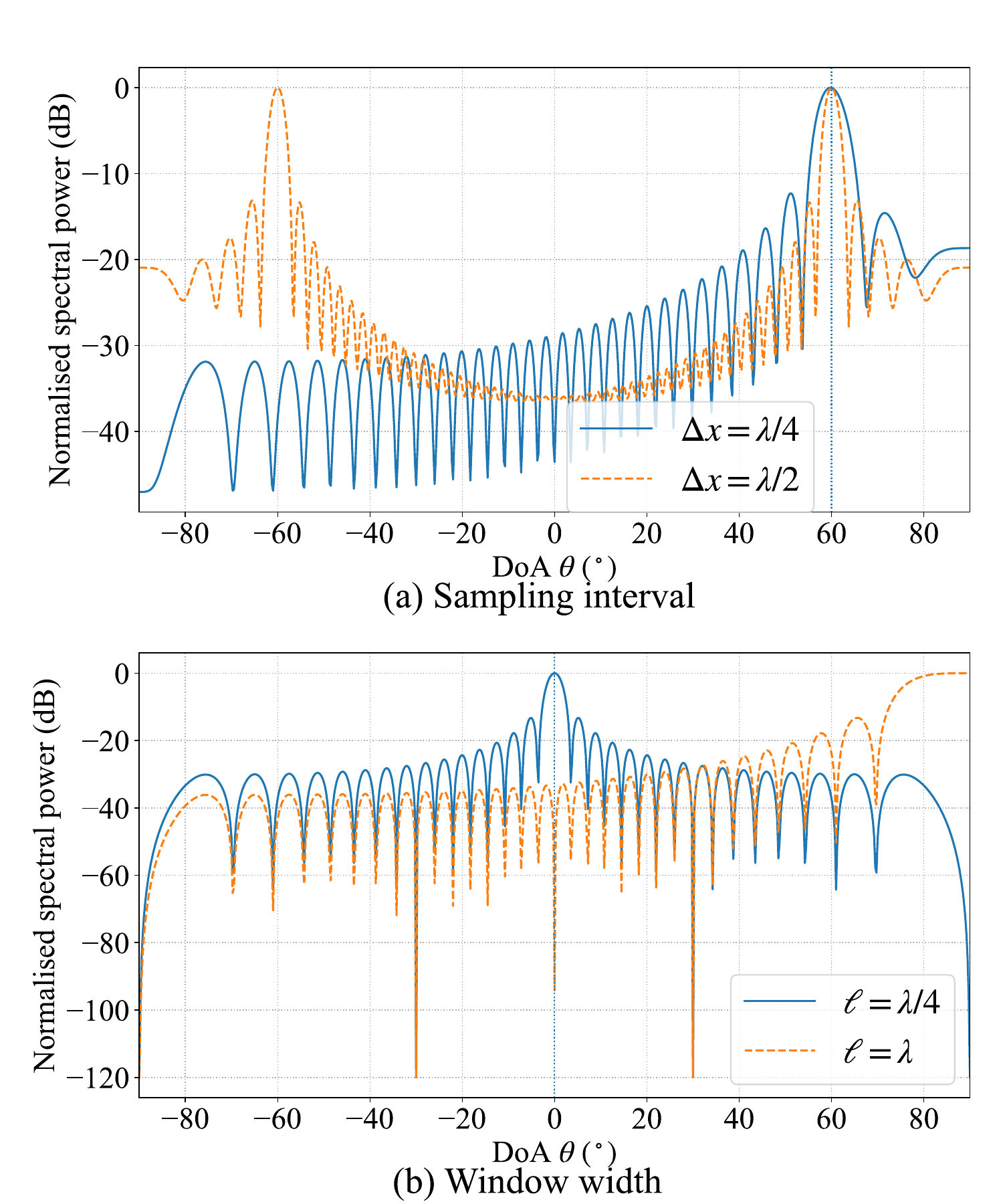}
\caption{Validation of the sampling theorem constraints for (a) sampling interval $\Delta x$ and (b) window width $\ell$.}
\label{fig:sampling_window}
\end{figure}
Fig.~\ref{fig:sampling_window} provides a visual confirmation of the sampling constraints derived in Theorem~\ref{thm:sampling}, which are critical for the unambiguous estimation of DoA. The plots show the normalized spectral power as a function of the angle of arrival, calculated from the signal model in \eqref{eq:sum-of-exps}.

Fig.~\ref{fig:sampling_window}(a) investigates the impact of the spatial sampling interval, $\Delta x$, which is the distance between the centers of the virtual measurement windows. For this simulation, a single signal with a true DoA of $60^\circ$ is used. Theorem~\ref{thm:sampling} establishes that to prevent spatial aliasing, this interval must satisfy $\Delta x \le \lambda/4$. The simulation compares two cases. When the sampling interval is set to $\Delta x = \lambda/4$, which adheres to the Nyquist criterion, the resulting spectrum  correctly identifies the signal's true direction with a clear peak at $60^\circ$. However, when the sampling interval is increased to $\Delta x = \lambda/2$, violating the criterion, the spectrum is corrupted by aliasing. This is shown by the dashed orange line, where the signal power is incorrectly mirrored to a false DoA near $-60^\circ$, demonstrating a complete failure of the estimation.

Fig.~\ref{fig:sampling_window}(b) illustrates the effect of the rectangular window width, $\ell$, by simulating a target at a DoA of $0^\circ$. The theorem dictates that the width must be less than half a wavelength, $\ell < \lambda/2$, to avoid creating nulls in the window's spectral response, which would lead to blind spots. When a narrow window width of $\ell = \lambda/4$ is used, the spectral response  correctly shows a peak at $0^\circ$, detecting the target. In contrast, when the window width is set to $\ell = \lambda$, a null in the window's Fourier transform  is created at the exact spatial frequency corresponding to the $0^\circ$ DoA. This creates a blind spot, completely suppressing the signal from the target, as seen by the deep null in the spectral power. This confirms that an improper window choice can make the sensor blind to signals from certain directions.

\section{Discussion}\label{sec:discussion}
Although the proposed ISE framework treats the fluorescence samples as ideal spatial measurements of $\alpha(x)$, implementing such sampling in a real Rydberg receiver may raise several hardware-level challenges. The first requirement is to obtain sufficient independent spatial samples along the vapor cell while respecting the aliasing constraints in Theorem~\ref{thm:sampling}. In practice, the spatial resolution is set by the imaging optics and the camera pixel size. Light-sheet fluorescence experiments with Rydberg vapors have demonstrated two-dimensional field imaging with $\sim 100\text{--}200~\mu\text{m}$ spatial resolution over centimeter-scale fields of view using standard sCMOS cameras \cite{schlossberger2024two}, which significantly exceeds the $\Delta x \le \lambda/4$ requirement at microwave wavelengths and allows for substantial oversampling. 

The LO-dominant regime assumed in our linearization requires a relatively strong and well-controlled RF LO field inside the cell. High-power microwave Rydberg experiments have shown that room-temperature vapor cells can tolerate strong fields while retaining a well-behaved spectroscopic response, provided that power broadening and transit-time effects are carefully managed \cite{anderson2016optical}. Combining these observations, existing fluorescence-based Rydberg imagers already satisfy most hardware requirements of ISE. Moreover, angle-of-arrival measurements based purely on fluorescence images of RF standing waves have recently achieved degree-level angular uncertainties at GHz carrier frequencies \cite{schlossberger2025angle}. These demonstrations indicate that implementing spatially resolved fluorescence sampling for multi-target DoA estimation is not limited by any fundamental hardware barrier, but rather by incremental engineering of the imaging optics, LO delivery, and camera readout chain.

\section{Conclusion}\label{V}
In this paper, we have presented a novel approach for multi-target DoA estimation using a single Rydberg atomic receiver, which addresses the limitations of complex receiver arrays and the single-target, narrow-band constraints of previous methods. The approach is based on spatially resolving the fluorescence profile within the vapor cell, which reduces the complex multi-target problem to a classical spectral estimation task. We demonstrated that under a strong LO-dominant regime, the atomic absorption pattern linearizes into a superposition of spatial sinusoids, where each spatial frequency directly maps to a target's DoA. This allows for the use of established algorithms like Prony's method to resolve multiple targets. Our ISE approach inherently supports multi-target detection and enables broadband operation by removing the restrictive dependency on a specific cell length. Simulations validated our theoretical model, showing near-CRLB performance and confirming that an LO-to-signal power ratio of at least 10 is important for high accuracy. This work offers an effective approach for high-precision, broadband DoA estimation and suggests a potential path toward advanced concepts such as multi-channel Rydberg receivers and the continuous-aperture sensing required for holographic MIMO. Future work will focus on the application of this approach in wireless communication systems.

\appendices
\section{Proof of Theorem~1}\label{Appen.A}
\begin{proof}
Under the LO–dominant linearization in Sec.~III-A, the calibrated channel outputs $\{\tilde y_j\}_{j=1}^K$ can be written as a sum of spatial sinusoids sampled on a 1-D grid,
\begin{equation}
    \tilde y_j
    = \Re\!\bigg\{\sum_{i=1}^N b_i e^{\mathrm{j}\Delta k_i x_j}\bigg\},
    \qquad x_j = x_0 + (j-1)\Delta x,
    \label{eq:app_sig_model}
\end{equation}
where $\Delta k_i$ is the spatial frequency associated with the $i$-th DoA and $b_i$ collects all direction–independent factors.

\subsubsection*{1) Constraint on $\Delta x$ }
Define the discrete–time index $n = j-1$ and the corresponding digital frequency
\[
    \omega_i \triangleq \Delta k_i \Delta x.
\]
Ignoring the irrelevant phase $e^{\mathrm{j}\Delta k_ix_0}$,
\eqref{eq:app_sig_model} is equivalent to
\begin{equation}
    s[n] = \sum_{i=1}^N c_i e^{\mathrm{j}\omega_i n},
    \qquad n = 0,\cdots,K-1,
\end{equation}
with $c_i = b_i e^{\mathrm{j}\Delta k_i x_0}$. In this discrete representation, spatial frequencies are only identifiable \emph{modulo} $2\pi$, i.e., the exponentials generated by
$\Delta k$ and
\[
    \Delta k' = \Delta k + \frac{2\pi q}{\Delta x},
    \qquad q \in \mathbb{Z},
\]
produce exactly the same sampled sequence because
$e^{\mathrm{j}(\omega + 2\pi q)n} = e^{\mathrm{j}\omega n}$ for
all integer $n$. Thus, to avoid aliasing, all physically admissible $\Delta k_i$ must lie inside a single Nyquist interval of width $2\pi/\Delta x$.

From the physical DoA model we have
\begin{equation}
    \Delta k_i = k\bigl(\sin\theta_0 - \sin\theta_i\bigr),
    \qquad k = \frac{2\pi}{\lambda}.
    \label{eq:app_deltak}
\end{equation}
For far-field incidence we restrict $\theta_i \in [-\pi/2,\pi/2]$, so that $\sin\theta_0 - \sin\theta_i \in [-2,2]$ and hence
\begin{equation}
    \Delta k_i \in [-2k,2k], \quad \forall i.
\end{equation}
Requiring the whole interval $[-2k,2k]$ to be contained within a single Nyquist interval $[-\pi/\Delta x,\pi/\Delta x]$ yields
\begin{equation}
    2k \le \frac{\pi}{\Delta x}
    \quad\Longrightarrow\quad
    \Delta x \le \frac{\pi}{2k}.
\end{equation}
Using $k = 2\pi/\lambda$ gives
\begin{equation}
    \Delta x \le \frac{\pi}{2\cdot 2\pi/\lambda}
            = \frac{\lambda}{4},
\end{equation}
which is the claimed sampling interval bound preventing spatial-frequency aliasing over the full DoA sector.

\subsubsection*{2) Constraint on $\ell$ }
Let the mother window $w_0(x)$ be a real rectangular window of width $\ell$ centered at the origin,
\begin{equation}
    w_0(x) =
    \begin{cases}
        1, & |x|\le \ell/2,\\
        0, & \text{otherwise}.
    \end{cases}
\end{equation}
Its spatial Fourier transform is well known,
\begin{align}
    \hat{w}_0(\omega)
    &= \int_{-\ell/2}^{\ell/2} e^{\mathrm{j}\omega x}\,\mathrm{d}x
     = \frac{2\sin(\omega\ell/2)}{\omega}                    \\
    &= \ell\,\mathrm{sinc}\!\Bigl(\frac{\omega\ell}{2}\Bigr),
\end{align}
so the non-zero zeros occur at
\begin{equation}
    \omega = \frac{2\pi m}{\ell}, \qquad m \in \mathbb{Z}\setminus\{0\}.
\end{equation}
The shifted windows $w_j(x) = w_0(x-x_j)$ used to form the virtual array have Fourier transforms
\[
    \hat{w}_j(\omega) = e^{\mathrm{j}\omega x_j}\hat{w}_0(\omega),
\]
i.e., they share the same \emph{magnitude} response and the same spectral zeros as $w_0$.

From the linearized model, the $i$-th target contributes an effective complex amplitude
\begin{equation}
    b_i = A_i\,\hat{w}_0(\Delta k_i)\,e^{-\mathrm{j}\Delta\phi_i},
\end{equation}
so $\hat{w}_0(\Delta k_i)$ plays the role of an element factor. If $\hat{w}_0(\Delta k_i) = 0$ for some admissible $\Delta k_i$, then \emph{all} channels receive zero contribution from that target, and the corresponding DoA becomes a blind direction.

We have already established that all spatial frequencies of interest satisfy $\Delta k_i \in [-2k,2k]$. To ensure that no non-zero zero of $\hat{w}_0(\omega)$ falls inside this band, the first zeros at $\omega = \pm 2\pi/\ell$ must lie strictly outside $[-2k,2k]$, i.e.,
\begin{equation}
    \frac{2\pi}{\ell} > 2k
    \quad\Longrightarrow\quad
    \ell < \frac{\pi}{k}.
\end{equation}
Substituting $k = 2\pi/\lambda$ yields
\begin{equation}
    \ell < \frac{\pi}{2\pi/\lambda} = \frac{\lambda}{2},
\end{equation}
which guarantees that the window’s spectral response has no nulls within the entire spatial-frequency band corresponding to the DoA sector. This completes the proof.
\end{proof}

\bibliographystyle{IEEEtran}
\bibliography{mybib}

\end{document}

%% file: defs.tex
\newtheoremstyle{noperiod}  
  {\topsep}                 
  {\topsep}                 
  {\itshape}                
  {}                        
  {\bfseries}               
  {}                        
  {.5em}                    
  {}                        

\theoremstyle{noperiod}     

\theoremstyle{noperiod}
\newtheorem{theorem}{Theorem}
\theoremstyle{noperiod}

\theoremstyle{noperiod}

\theoremstyle{noperiod}

\theoremstyle{noperiod}

